\documentclass{llncs}
\usepackage{amssymb}
\usepackage{amsmath}
\usepackage{graphicx}
\usepackage[justification=centering]{caption}
\usepackage{mathptmx}      
\newcommand{\Rset}{\mathbb{R}}
\newcommand{\Dscr}{\mathcal{D}}
\newcommand{\Hscr}{\mathcal{H}}
\newcommand{\Pscr}{\mathcal{P}}
\begin{document}
\title{The Salesman's Improved Tours for Fundamental Classes} 
\author{Sylvia Boyd\thanks{SEECS, University of Ottawa, Ottawa, Canada.  Partially supported by the National Sciences and Engineering Research Council of Canada.  This work was done during visits in  Laboratoire G-SCOP, Grenoble; support from the CNRS and   Grenoble-INP is gratefully acknowledged.} and  
Andr\'as Seb\H{o}\thanks{CNRS, Laboratoire G-SCOP, Univ.~Grenoble Alpes, Supported by LabEx PERSYVAL-Lab (ANR 11-LABX-0025), \'equipe-action GALOIS.}}



\institute{}

\begingroup
\makeatletter
\let\@fnsymbol\@arabic
\maketitle
\endgroup

\begin{abstract}
Finding the exact integrality gap $\alpha$ for the LP relaxation of the metric Travelling Salesman Problem (TSP) has been an open problem for over thirty years, with little progress made.  It is known that $4/3 \leq \alpha \leq 3/2$, and a famous conjecture states $\alpha = 4/3$.  It has also been conjectured that there exist half-integer basic solutions of the linear program for which the highest integrality gap is reached.  

For this problem, essentially two ``fundamental'' classes of instances have been proposed.  This  fundamental property means that in order to show that the integrality gap is at most $\rho$ for all instances of the metric TSP, it is sufficient to show it only for the instances in the fundamental class.  However, despite the importance and the simplicity of such classes, no apparent effort has been deployed for improving the integrality gap bounds for them.
In this paper we take a natural first step in this endeavour, and consider the $1/2$-integer points of one such class.  We successfully improve the upper bound for the integrality gap from $3/2$ to $10/7$ for a superclass of these points for which a lower bound of $4/3$ is proved.

A key role in the proof of this result is played by finding Hamiltonian cycles whose existence is equivalent to Kotzig's result on "compatible Eulerian tours", and which lead us to delta-matroids for developing the related algorithms.  Our arguments also involve other innovative tools from combinatorial optimization with the potential of a broader use. 
 
\end{abstract}

\section{Introduction}\label{Intro}
Given the complete graph $K_n = (V_n, E_n)$ on $n$ nodes with non-negative edge costs $c
\in {\mathbb R}^{E_n}$, the \emph{Traveling Salesman Problem}
(henceforth TSP) is to find a Hamiltonian cycle of minimum cost in
$K_n$.  When the costs are \emph{metric}, i.e. satisfy the triangle inequality
$c_{ij} + c_{jk} \geq c_{ik}$ for all $i,j,k \in V_n$, the problem
is called the \emph{metric} TSP. If the metric is defined by the shortest (cardinality) paths of a graph, then it is called a {\em graph metric};  the TSP specialized to graph metrics  is the {\em graph TSP}.

For $G=(V,E)$, $x \in {\mathbb R}^E$ and $F \subseteq E$,
$x(F):=\sum_{e \in F} x_e$; for $U \subseteq
V$, $\delta(U):=\delta_G(U):=\{uv \in E : u \in U, v \in V\setminus U\}$; $E[U]:=\{uv \in E : u \in U, v \in U\}$. The scalar product of vectors $a$ and $x$ of the same dimension will simply be denoted by $ax$. A {\em path} is the edge set of a connected subgraph with two nodes of degree $1$ and all other nodes of degree $2$, and a {\em cycle} is the edge set of a connected subgraph with all node degrees equal to $2$. 

A natural linear programming relaxation for the TSP is the following \emph{subtour LP}:
\begin{eqnarray}
\mbox{minimize } cx \label{oby}\\
\mbox{subject to:  } x(\delta(v)) = 2 &\quad \mbox{ for all } v \in
V_n,\label{deg}\\
x(\delta(S)) \geq 2 &\quad \mbox{ for all } \emptyset\neq S\subsetneq V_n, \label{st}\\
0 \leq x_e \leq 1 &\quad \mbox{for all } e \in E_n. \label{ul}
\end{eqnarray}

\noindent For a given cost function $c\in \Rset^{E_n}$, we use $LP(c)$ to denote the optimal solution value for the subtour LP and $OPT(c)$ to denote the optimal solution value for the TSP. The polytope associated with the subtour LP, called the \emph{subtour elimination polytope} and denoted
by $S^n$, is the set of all vectors $x$ satisfying the constraints of the subtour LP,
i.e. 
$S^n = \lbrace x \in \Rset^{E_n} : x \mbox{ satisfies } (\ref{deg}), (\ref{st}), (\ref{ul}) \rbrace$.

The metric TSP is known to be NP-hard.  One approach taken for finding reasonably good solutions is to
look for a {\em $\rho$-approximation
	algorithm} for the problem, i.e.  a polynomial-time algorithm that always computes a solution
of value at most $\rho$ times the optimum. Currently the best such
algorithm known for the metric TSP is the algorithm due to Christofides \cite{chr} for which
$\rho = \frac{3}2$.  Although it is widely believed that a better approximation algorithm is possible, no one has been able to improve upon Christofides algorithm 
in four decades. For arbitrary nonnegative costs not constrained by the triangle inequality there does not
exist a $\rho$-approximation algorithm for any $\rho\in\mathbb{R}$ unless $P=NP$, since such an algorithm would be able to decide if a given graph is Hamiltonian.

For an approximation guarantee of a minimization problem one needs lower bounds for the optimum, often provided by linear programming.  For the metric TSP with cost function $c$, a commonly used lower bound is $LP(c)$. Then finding a solution, i.e. a Hamiltonian cycle, of objective value at most $\rho\,LP(c)$ in polynomial time implies at the same time a $\rho$-approximation algorithm, and establishes that the {\em integrality gap}  $OPT(c)/LP(c)$  is at most $\rho$ for any input.  (The input consists of the nodes and the metric function $c$ on pairs of nodes; again, without the metric assumption this ratio is unbounded already for the graph TSP by putting infinite costs on non-edges of the defining graph.)  Since up until now the bounds on the integrality gap have been proved via polynomial algorithms constructing the Hamiltonian cycles, the approximation ratio for metric costs is conjectured not to be larger than the integrality gap.  

It is known that the integrality gap for the subtour LP for any metric $c$ is at most $\frac{3}2$ (\cite{cunningham}, \cite{SW}, \cite{Wolsey80}),
however no example for
which the integrality gap is greater than $\frac{4}3$ is known.  In fact, a famous conjecture, often referred to as the \emph{$\frac{4}3$ Conjecture},
states the following:

\begin{conjecture} \label{conj:4over3}
	The integrality gap for the subtour LP with metric $c$ is at most $\frac{4}3$.
\end{conjecture}
\noindent Well-known examples have an integrality gap asymptotically equal to $\frac{4}3$. 
In almost thirty years, there have been no improvements made to the upper bound of $\frac{3}2$ or lower bound of $\frac{4}3$ for the integrality gap for the subtour LP with metric $c$.

Define a \emph{tour} to be the edge set of a spanning Eulerian (i.e. connected with all degrees even) multisubgraph of $K_n$.  If none of the multiplicities can be decreased, then all multiplicities are at most two; however, there are some technical advantages to allowing higher multiplicities. 

For any multiset $J\subseteq E_n$, the {\emph {incidence vector of $J$}, denoted by $\chi^J$, is the vector in $\Rset^{E_n}$ for which $\chi^{J}_e$ is equal to the number of copies of edge $e$ in $J$ for all $e \in E_n$. We use $T^n$ to denote the convex hull of incidence vectors of tours of $K_n$, and for costs $c\in \Rset^{E_n}$ we use $OPT_{T^n}(c)$ to denote the cost of a minimum cost tour.  Note that $T^n$ is an unbounded polyhedron, as $T^n + \Rset^{E_n}_+ = T^n$: each edge may have arbitrarily large multiplicity.



	For any  $\rho \in {\mathbb R}$,  $\rho\,S^n$ denotes $\{y \in {\mathbb R}^{E_n} : y = \rho x,  x \in S^n \}$. The definition of the integrality gap can be reformulated in terms of a containment relation between the two polyhedra $\rho S^n$ and $T^n$ (Theorem~\ref{thm:vectorform}) that does not depend on the objective function. We not only use this reformulation here, but also develop a specific way of exploiting  it, and for our arguments this is the very tool that works.  Showing for some constant $\rho\in\mathbb{R}$ that $\rho\,x \in T^n$ for each $x\in S^n$, i.e. that $\rho\,x$ is a convex combination of incidence vectors of tours, gives an upper bound of $\rho$ on the integrality gap for the subtour LP:  it implies that for each $x\in S^n$ and any cost function $c\in \Rset^{E_n}$ such that $cx = LP(c)$, at least one of the tours in the convex combination has cost at most $\rho\,(cx) = \rho\,LP(c)$. If the costs are metric, this tour can be  shortcut to a TSP solution of cost at most $\rho\,LP(c)$, giving a ratio of $OPT(c)/LP(c) \leq \rho$. A {\em shortcut} means to fix an Eulerian tour and replace a sequence of nodes $a,b,c)$ by $a,c$, whenever $b$ has already been visited by the tour.  The essential part ``(i) implies (ii)'' of the following theorem, due to Goemans \cite{goemans} (also see \cite{CV}), asserts that the converse is also true: if $\rho$ is at least the integrality gap then  $\rho\,S^n$ is a subset of $T^n$.  
	
	\begin{theorem}\label{thm:vectorform}\cite{goemans}\cite{CV}   Let $K_n = (V_n, E_n)$ be the complete graph on $n$ nodes and let $\rho\in\Rset, \rho\ge 1$. The following statements are equivalent: 
		\begin{itemize}
			\item[(i)] For any metric cost function $c:E_n\longrightarrow \Rset_+, \mbox{     }OPT(c)\le \rho LP(c)$.
			\item[(ii)] For any $x\in S^n$, $\rho x \in T^n$. 
			\item[(iii)] For any vertex $x$ of $S^n$, $\rho x \in T^n$. 	
		\end{itemize}
	\end{theorem}

	  By Theorem \ref{thm:vectorform},  Conjecture \ref{conj:4over3} can be equivalently  reformulated as follows:
	\begin{conjecture}\label{conj:fourthird}   The polytope $\frac{4}{3}\,S^n$ is contained in the polyhedron $T^n$,  that is, $\frac{4}{3}\,S^n \subseteq T^n$.
	\end{conjecture}
	
	Given a vector $x \in S^n$, the \emph{support graph $G_x=(V_n, E_x)$ of $x$} is defined with $E_x = \{e\in E_n:  x_e >0 \}$.  We call a point $x \in S_n$ \emph{$\frac{1}2$-integer} if $x_e \in \{0, \frac{1}2, 1\}$ for all $e \in E_n$. For such a vector we call the edges $e \in E_n$ \emph{$\frac{1}2$-edges} if $x_e = \frac{1}2$ and \emph{$1$-edges} if $x_e = 1$.  Note that the $1$-edges of $\frac{1}2$-integer points form a set of disjoint paths that we call \emph{$1$-paths} of $x$, and the $\frac{1}2$-edges form a set of edge-disjoint cycles we call the \emph {$\frac{1}2$-cycles} of $x$. 
	
	For Conjecture \ref{conj:4over3}, it seems that $\frac{1}2$-integer vertices play an important role {(see \cite{BB},\cite{CR},\cite{SWZ}).  In fact it has been conjectured by Schalekamp, Williamson and van Zuylen \cite{SWZ} that a subclass of these $\frac{1}2$-integer vertices are the ones that give the largest gap. Here we state a weaker version of their conjecture:
		
		\begin{conjecture}\label{conj:half} 
			The integrality gap for the subtour LP is reached on $\frac{1}2$-integer vertices.
		\end{conjecture}
		
		Very little progress has been made on the above conjectures, even though they have been around for a long time and have been well-studied. For the special case of graph TSP an upper bound of $\frac{7}5$ is known for the integrality gap \cite{AndrasJens}.  Conjecture \ref{conj:fourthird} has been verified for the so-called \emph{triangle vertices} $x\in S^n$ for which the values are $\frac{1}2$-integer, and the $\frac{1}2$-edges form triangles in the support graph \cite{BC}.  
		The lower bound of $\frac{4}3$ for the integrality gap is provided  by triangle vertices with just two triangles.

		
		A concept first introduced by Carr and Ravi \cite{CR} (for the $2$-edge-connected subgraph problem) is that of a \emph{fundamental class}, which is a class of points $F$ in the subtour elimination polytopes $S^n$, $n\ge 3$ with the following property:  showing that $\rho\,x$ is a convex combination of incidence vectors of tours for all vertices $x \in F$ implies the same holds for \emph{all} vertices of polytopes $S^n$, and thus implies that the integrality gap for the subtour LP is at most $\rho$.  
		
		Two main classes of such vertices have been introduced, one by Carr and Vempala  \cite{CV}, the other by Boyd and Carr \cite{BC}.   In this paper we will focus on the latter one, that is,  we define a \emph{Boyd-Carr point} \cite{BC} to be a point $x \in S^n$ that satisfies the following conditions:
		
		\begin{itemize}
			\item[(i)] The support graph $G_x$ of $x$ is cubic.
			\item[(ii)] In $G_x$, there is exactly one $1$-edge incident to each node.
			\item[(iii)] The fractional edges of $G_x$ form disjoint $4$-cycles.
		\end{itemize}
		
		A \emph{Carr-Vempala point} \cite{CV} is one that satisfies (i), (ii) and instead of (iii), the fractional edges form a Hamiltonian cycle. We use {\em fundamental point} as a common name for points that are either Boyd-Carr or Carr-Vempala points. It has been proved that the Boyd-Carr points \cite{BC} and Carr-Vempala points \cite{CV} each form fundamental classes. The support graph of a fundamental point will be called a \emph{fundamental graph}.  In other words, a fundamental graph is a cubic graph where there exists a perfect matching whose deletion leads to a graph whose components are $4$-cycles, or  a Hamiltonian cycle.  Note that the  $3$-edge-connected instances of each of the two classes of fundamental points also form fundamental classes: for Boyd-Carr points this can be checked from the construction itself \cite{BC}; for Carr-Vempala points this is obvious, since the Hamiltonian cycle of edges $\{e\in E_n: x(e)<1\}$ has at least two edges in every cut, but two such edges alone do not suffice for constraints (\ref{st}) of the subtour LP. 
		
		Despite their significance and simplicity, no effort has been deployed to exploring new integrality gap bounds for these classes, and no improvement on the general $\frac{3}{2}$ upper bound on the integrality gap has been made for them, not even for special cases. A natural first step in this endeavour is to try to improve the general bounds for the special case of $\frac{1}2$-integer Boyd-Carr or Carr-Vempala points.
		
		In this paper we improve the upper bound for the integrality gap from $\frac{3}2$ to $\frac{10}7$ for $\frac{1}2$-integer Boyd-Carr points, and also provide a $\frac{10}7$-approximation algorithm for metric TSP for any cost function for which the subtour LP is optimized by a Boyd-Carr point.  In fact we generalize these results to a superclass of these points.  Replacing the $1$-edges in Boyd-Carr points by $1$-paths of arbitrary length between their two endpoints, we get all the  $\emph{square points}$, that is,  $\frac{1}2$-integer points of $S^n$ for which the $\frac{1}2$-edges form disjoint  $4$-cycles, called \emph{squares} of the support graph. 
		We also show that there exists a subclass of square points that provide instances where the integrality gap is at least $\frac{4}3$. 
	    Thus square points provide new tight examples for the lower bound of the conjectures. 
		
		In the endeavour to find improved upper bounds on the integrality gap we examine the  structure of the support graphs of Boyd-Carr  points. We show that they are all Hamiltonian, an important ingredient of our bounding of their integrality gap. The proof uses a theorem of Kotzig \cite{Kotzig}  on  Eulerian trails with forbidden transitions. An {\em Eulerian trail} in a graph is a closed walk containing each of its edges exactly once. Contrary to tours, it is more than just an edge set, the order of the edges also plays a role.  The connection of tours to Eulerian trails leads  us to delta-matroids and to developing related algorithms, which are discussed in Section \ref{sec:delta}. 
		
		
		In Section~\ref{sub:forbidden} we show a first, basic application of these ideas, where some parts of the difficulties do not occur. We  prove that all edges can be uniformly covered $6/7$ times by tours in the support graphs of both fundamental classes in the case where they are $3$-edge-connected. This is better than the conjectured general bound $8/9$ that would follow for arbitrary $3$-edge-connected cubic graphs from Conjecture~\ref{conj:fourthird} (see \cite{Sebo14}).

		
		
		Another new way of using classical combinatorial optimization for the TSP occurs in Section~\ref{sub:rainbow}, where we use an application of Edmonds' matroid intersection theorem to write the optimum $x$ of the subtour elimination polytope as the convex combination of incidence vectors of ``rainbow'' spanning trees in edge-coloured graphs.  The idea of using spanning trees with special structures to get improved results has recently been used successfully in \cite{G} for graph TSP, and in \cite{JC} and \cite{AA} for a related problem, namely the metric $s-t$ path TSP.  However, note that we obtain and use our trees in a completely different way. 
		
		Our main results concerning the integrality ratio of $\frac12$-integer Boyd-Carr points and square points are proved in Section~\ref{main}. We conclude that section by outlining a potential strategy for using the Carr-Vempala points of  \cite{CV} for proving the $\frac{4}3$ Conjecture.  
		
		Finally, in Section \ref{sec:delta}, we provide polynomial-time optimization algorithms for some of the existence theorems of previous sections, including a $\frac{10}7$-approximation algorithm for metric TSP for any cost function which is optimized by a square point for the subtour LP.  The methods used relate to delta matroids, and their relevance is discussed.
		
		\section{Polyhedral Preliminaries and Other Useful Tools}\label{tools}
		
		In this section we discuss some useful and powerful tools that we  need in the proof of our main result in Section~\ref{main}. We begin with some preliminaries.  
		
		Given a graph $G=(V,E)$ with a node in $V$ labelled $1$, a \emph{$1$-tree} is a subset $F$ of $E$ such that $|F\cap \delta(1)| = 2$ and $F\backslash\delta(1)$ forms a spanning tree on $V\backslash\{1\}$.  The convex hull of the incidence vectors of $1$-trees of $G$, which we will refer to as the \emph {$1$-tree polytope} of the graph $G$, is given by the following \cite[page 262]{theTSPbook}:
		\begin{gather}
		\{x\in \Rset^E:	\mbox{    }x(\delta(1)) = 2, \mbox{    }
		x((E[U])) \leq |U| - 1 \mbox{  for all   } \emptyset \neq U \subseteq V\backslash\{1\},\nonumber\\[1ex]
		0 \leq x_e \leq 1\mbox{   for all   } e \in E \mbox{,     } x(E)=|V|\}.\label{onetreepoly}
		\end{gather}
		It is well-known that the $1$-trees of a connected graph satisfy the basis axioms of a matroid (see \cite[page 262-263]{theTSPbook}).
		
		Given $G=(V,E)$ and $T\subseteq V$, $|T|$ even, a \emph {$T$-join} of $G$ is a set $J\subseteq E$ such that $T$ is the set of odd degree nodes of the graph $(V,J)$.  A cut $C=\delta(S)$ for some $S\subset V$ is called a \emph{$T$-cut} if $|S\cap T|$ is odd. We say that a vector {\em majorates} another if it is coordinatewise greater than or equal to it. The set of all vectors $x$ that majorate some vector $y$ in the convex hull of incidence vectors of $T$-joins of $G$ is given by the following \cite{EdmondsJ73}: 
		\begin{eqnarray}\label{Tjoinpoly}
		\{x\in \Rset^E: x(C) \geq 1 \mbox{ for each } T\mbox{-cut C, } x_e \geq 0 \mbox{ for all }e \in E\}.
		\end{eqnarray} 
		\noindent This is the \emph {$T$-join polyhedron} of the graph $G$. 
		
		The following two results are well-known (see \cite{cunningham}, \cite{SW}, \cite{Wolsey80}), but we include the proofs as they introduce the kind of polyhedral arguments we will use: 
		\begin{lemma}\label{lem:Wolsey} \cite{cunningham} \cite{SW} \cite{Wolsey80}
			If $x\in S^n$, then (i) it is a convex combination of incidence vectors of $1$-trees of $K_n$, and (ii) $x/2$ majorates a convex combination of incidence vectors of $T$-joins of $K_n$ for every $T\subseteq V_n$, $|T|$ even. 
		\end{lemma}
		
		\begin{proof} Constraints (\ref{deg}) and (\ref{st}) of the subtour LP together imply that $x(E_n) = n$ and $x(E_n[S])\leq |S|-1$, for all $\emptyset\neq S\subsetneq V_n$.  Thus $x\in S^n$ satisfies all of the constraints of the $1$-tree polytope of $K_n$ and (i) of the lemma follows.  To check (ii), note that for all $T\subseteq V_n$, $|T|$ even, $x/2$ satisfies the constraints of the $T$-join polyhedron of $K_n$  (in fact $x(C)/2\ge 1$  for every cut $C$),  that is, it majorates a convex combination of incidence vectors of $T$-joins.  \qed
		\end{proof}
		
		\begin{theorem}\label{thm:Wolsey} \cite{cunningham} \cite{SW} \cite{Wolsey80}
			If $x\in S^n$, $\frac32 x \in T^n$. 
		\end{theorem}		
		
		\begin{proof} By (i) of Lemma~\ref{lem:Wolsey}, $x$ is a convex combination of incidence vectors of $1$-trees of $K_n$. Let $F$ be any $1$-tree of such a convex combination, and $T_F$ be the set of odd degree nodes in the graph $(V_n,F)$. Then by (ii) of Lemma~\ref{lem:Wolsey}, $x/2$ majorates a convex combination of incidence vectors of $T_F$-joins. 
		So $\chi^F + x/2$ majorates a convex combination of incidence vectors of tours, and taking the average with the coefficients of the convex combination of $1$-trees, we get that $x+x/2$ majorates a convex combination of incidence vectors of tours. Since adding $2$ to the multiplicity of any edge in a tour results in another tour, it follows that $\frac32 x \in T^n$. \qed \end{proof}
		
		
		The tools of the following two subsections are new for the TSP and appear to be very useful. 
		\subsection{Eulerian Trails with Forbidden Bitransitions}\label{sub:forbidden} 
		
		\noindent
		Let $G=(V,E)$ be a connected $4$-regular multigraph. For any node $v\in V$, a \emph{bitransition} (at $v$) is a partition of $\delta(v)$ into two pairs of edges.  Clearly every Eulerian trail of $G$ {\em uses} exactly one bitransition at every node, meaning the two disjoint pairs of consecutive edges of the trail at the node.  There are three bitransitions at every node and the theorem below, equivalent to a result of Kotzig \cite{Kotzig}, states that we can forbid one of these and still have an allowed Eulerian trail.  As we will show, this provides Hamiltonian cycles in the support of square points, used in Section~\ref{main} to prove our main result. 
		
		\begin{theorem}\cite{Kotzig}\label{thm:Kotzig}
			Let $G=(V,E)$ be a 4-regular connected multigraph with a forbidden bitransition 
			for every $v\in V$. Then $G$ has an Eulerian trail not using the forbidden bitransition of any node. 
		\end{theorem}
	A \emph{square graph} is defined as a pair $(G,M)$ where $G=(V,E)$ is a cubic $2$-edge-connected graph, and $M$ is a perfect matching of $G$ such that the edges  $E\setminus M$ form squares.  We associate Boyd-Carr points to square graphs, where $M$ is defined to be the set of $1$-edges.
		
		\begin{lemma}\label{lem:Ham} A square graph $(G,M)$ has a Hamiltonian cycle containing $M$. 
		\end{lemma}
		
		\begin{proof} Let $G=(V,E)$. We assume $G$ has at least two squares as the lemma is trivially true otherwise.
			
					Suppose first that $G$ has a square $u_1u_2,u_2u_3,$ $u_3u_4,u_4u_1$ in $E\setminus M$, with a chord, say $u_2u_4$, in $M$, that is, exactly the two edges  $au_1$ and $u_3b$ $(a,b\in V)$ are leaving the set $\{u_1,u_2,u_3,u_4\}\subseteq V$.   In this case we can delete $u_1,u_2,u_3,u_4$ from the graph and add the edge $ab$ to $G$ to form a new square graph $(\hat{G},\hat{M})$ with $\hat{M} = M - u_2u_4 + ab$. From a Hamiltonian cycle of this reduced graph containing $\hat{M}$ we obtain a Hamiltonian cycle of $G$ containing $M$ by deleting $ab$, and adding $au_1,u_1u_2,u_2u_4,u_4u_3,u_3b$. We can thus suppose  that $G$ has no such square.  
			
		Contracting all squares of $G$, we obtain a $4$-regular connected multigraph $G^\prime=(V^\prime,E^\prime)$ whose edges are precisely $M$ and whose nodes are precisely the squares of $G\setminus M$. To each contracted square $C$ 
		we associate the forbidden bitransition consisting of the pairs of  edges of $M$ incident with the diagonally opposite nodes of $C$ in $G$, as shown in Figure \ref{fig:shrinksquare}.  By Theorem \ref{thm:Kotzig}, there is an Eulerian trail $K$ of $G^\prime$ that does not use these forbidden bitransitions.  The two pairs of consecutive edges in $K$ at each node in $G^\prime$ can then be completed by a perfect matching of the corresponding square in $G$, forming the desired Hamiltonian cycle.\qed	
		\end{proof}

		\begin{figure}
			\centering
			\includegraphics[scale=.60]{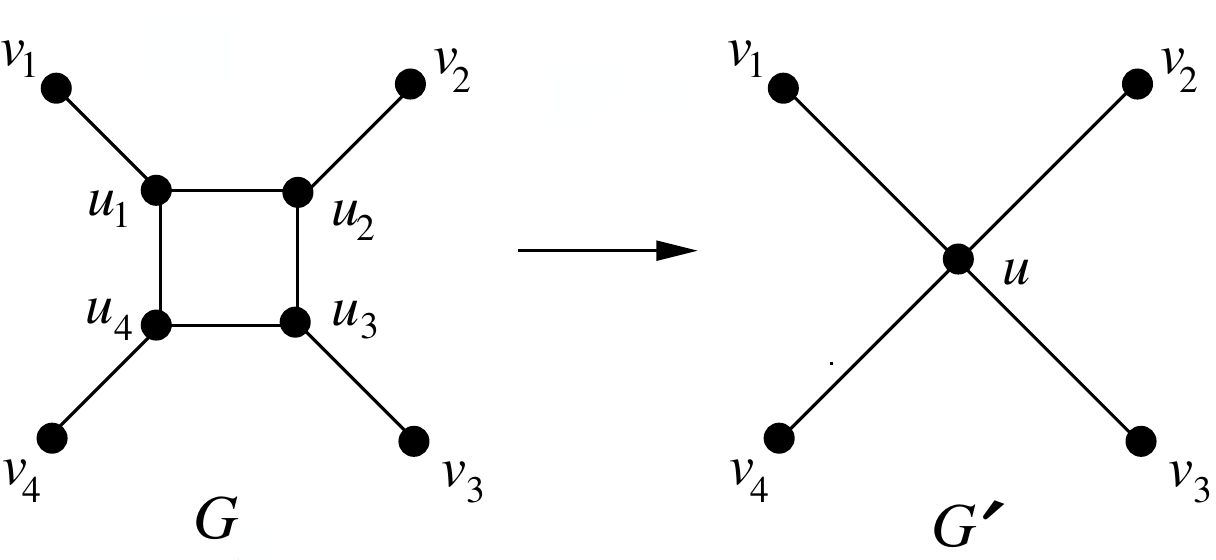}
			\caption{Shrinking a square in $G$ to  node $u$; forbidden: $\{(uv_1,uv_3),(uv_2,uv_4)\}$.}
			\label{fig:shrinksquare}
		\end{figure}

		The exhibited connection of Eulerian graphs with forbidden bitransitions sends us to a link on delta-matroids \cite{BC95} with well-known optimization properties. We explore this link in Section~\ref{sec:delta}, where we provide a direct self-contained algorithmic proof (with a polynomial-time, greedy algorithm) for a weighted generalization of Lemma \ref{lem:Ham}. We content ourselves in this section by providing a simple, first application of Lemma~\ref{lem:Ham} which shows a basic idea we will use in the proof of our main result in Section~\ref{main}, without the additional difficulty of the more refined application.     
		
		Given a graph $G=(V,E)$ and a value $\lambda$, the \emph{everywhere $\lambda$ vector for $G$} is the vector $y \in \Rset^{E_{|V|}}$ for which $y_e=\lambda$ for all edges $e\in E$ and $y_e=0$ for all the other edges in the complete graph $K_{|V|}$.  In Theorem~\ref{thm:everywherevector} below we show that for any cubic $3$-edge-connected graph with a Hamiltonian cycle, the everywhere $6/7$ vector is in $T^n$.  Since fundamental graphs are Hamiltonian (by Lemma~\ref{lem:Ham} for Boyd-Carr, and by definition for Carr-Vempala), the theorem applies to their $3$-edge-connected instances. Both classes of such $3$-edge-connected graphs are also fundamental classes as was noted earlier (see the remark after the definition).  
		
		\begin{theorem}\label{thm:everywherevector} If $G=(V,E)$ is a cubic, $3$-edge-connected Hamiltonian graph, then the everywhere $6/7$ vector for $G$ is in $T^n$.   
		\end{theorem}
		
		\begin{proof}
			Let $H$ be a Hamiltonian cycle of G, and let $M:=E\setminus H$ be the perfect matching complementary to $H$.  Note that $\chi^H$ is the incidence vector of a tour, and we will use it in the convex combination for the everywhere $6/7$ vector.  It is a good choice in that it is a tour which contains the fewest edges possible, however it will need to be balanced with tours that do not use the edges in $H$ very often in order to achieve our goal.  To this end we consider the point $x \in\Rset^{E_{|V|}}$ defined by $x_e=1$ if $e\in M$, $x_e=1/2$ if $e\in H$ and $x_e=0$ otherwise.  It is easily seen that $x$ is in the subtour elimination polytope $S^{|V|}$, thus by Theorem~\ref {thm:Wolsey},  $\frac32 x\in T^n$. 
			
			Now take the convex combination of tours $t:=\frac37\chi^H + \frac47 \frac32 x$. Then for edges $e\in M$ we have $t_e = 0 + \frac47\frac32=\frac67$. For edges $e\in H$ we have $t_e =\frac37 + \frac47\frac32\frac12=\frac67$, and $x_e = 0$ for all edges $e$ not in $G$, finishing the proof.   \qed
		\end{proof}
		
	    Replacing  Hamiltonian cycles by other, relatively small tours or convex combinations of tours in the proof of Theorem~\ref{thm:everywherevector}, one may  get similar results weakening the Hamiltonicity condition. Such results are particularly interesting for general, cubic $3$-edge-connected graphs.  For such graphs  the everywhere $1$ vector is in $T^n$, as noticed in \cite{Sebo14},  where it is then asked  whether the everywhere $\frac89$ vector belongs to $T^n$. These are the values one gets from Theorem~\ref{thm:Wolsey} and Conjecture~\ref{conj:fourthird} respectively, applied to the everywhere $\frac23$ vector,  feasible for $S^n$. Note that the analogous problem for the $s-t$ path TSP has been solved \cite{AA}. 
	     
	      The above possibility has been  explored by Haddadan, Newman and Ravi \cite{HNR}, who proved that the everywhere $18/19$ vector is in $T^n$, getting this constant below $1$ for the first time, and with a polynomial-time algorithm.    They replaced the Hamiltonian cycle in the proof of Theorem~\ref{thm:everywherevector} by a convex combination of tours proved by Kaiser and \v{S}krekovski \cite{KS} and by Boyd, Iwata, Takazawa \cite{BIT} algorithmically. Using this convex combination they also deduce better bounds with simple proofs for node-weighted graphs. Let us do the same for node-weighted Hamiltonian graphs. 
		
		In a {\em node-weighted TSP} every node $v$ of a given graph $G=(V,E)$ is given a weight $f_v \in \Rset_+$, and the cost $c_{uv}$ of an edge $uv \in K_{|V|}$ is $f_u + f_v$ if $uv \in E$, and the cost of a shortest path in $G$ otherwise. Note that $c$ is metric.  
		
		\begin{theorem}
			Let $G=(V,E)$ be a node-weighted cubic $3$-edge-connected graph for which the everywhere $\lambda$ vector $y$ is in $T^n$.  Then $OPT(c) \le \frac32 \lambda LP(c)$.
		\end{theorem}
	
	\begin{proof}
		Since $G$ is cubic we have $c(E)=3 \sum (f_v: v \in V)$.  Since $y\in T^n$ there is a tour $J$ such that 
		\begin{equation} \label{cheapJ}
		c(J) \le \lambda c(E) = 3\lambda \sum (f_v: v \in V).
	\end{equation} 
		As observed in \cite{HNR},  it follows from  constraints (\ref{deg}) of the subtour LP and the fact that the everywhere $\frac23$ vector on $E(G)$ is feasible for $S^{|V|}$ that $LP(c) = 2 \sum (f_v: v\in V)$.   
		Together with (\ref{cheapJ}) 
		and the fact that $OPT(c) \le c(J)$ for metric costs completes the proof.\qed
\end{proof}

\begin{corollary}
If $G=(V,E)$ is a node-weighted cubic, $3$-edge-connected Hamiltonian graph, or in particular a $3$-edge-connected fundamental graph, then $OPT(c) \le \frac97 LP(c)$.
\end{corollary}

		\subsection{Rainbow $1$-trees}\label{sub:rainbow}
		We now use   matroid intersection   to prove that not only is $x$ is in the convex hull of incidence vectors of $1$-trees, but we can also require that these $1$-trees satisfy some additional useful properties. We will use this in the proof of our main result in Section \ref{main}. 
		
		Given a graph $G=(V,E)$, let every edge of $G$ be given a colour. We call a $1$-tree $F$ of $G$ a \emph{rainbow $1$-tree} if every edge of $F$ has a different colour.  Rainbow  trees are discussed by Broersma and Li in \cite{BL}, where they note they are the common independent sets of two matroids, a fact  combined here with a polyhedral argument to obtain the following theorem:  
		\begin{theorem}\label{thm:michel}
	Let $x\in S^n$ be $\frac{1}2$-integer, and let $\Pscr$ be any partition of the $\frac{1}2$-edges into pairs.  Then  $x$ is in the convex hull of incidence vectors of $1$-trees that each contain exactly one edge from each pair in $\Pscr$.
\end{theorem}

\begin{proof}	
	Let $G_x = (V_n, E_x)$ be the support graph of $x$.  Consider the partition matroid (cf. \cite{Schrijver}) defined on $E_x$  by the partition $\Pscr \cup \{\{e\} : e\in E_x, \hbox{ $e$ is a $1$-edge} \}$. By Lemma~\ref{lem:Wolsey}, $x$ is in the convex hull of incidence vectors of $1$-trees in $E_x$; since   $x(Q)=1$ for every class $Q$ of the defined partition matroid, it is also in the convex hull of its bases.  Thus by \cite[Corollary 41.12d]{Schrijver}, which is a corollary of Edmonds' matroid intersection theorem \cite{EdmondsMIpolytope}, $x$ is in the convex hull of incidence vectors of the common bases of the two matroids.  \qed 
\end{proof}


\section{Improved Bounds for $\frac12$-integer Points} \label{main}
In this section we show that $\frac{10}7 x \in T^n$ for all square points $x\in S^n$, and thus for all $\frac{1}{2}$-integer Boyd-Carr points $x$ as well. We also analyse the possibility of a similar proof for Carr-Vempala points that would have the advantage of also implying a ratio of $\frac{10}7$ for all $\frac{1}{2}$-integer points in $S^n$, as we will discuss at the end of this section. We begin by stating two properties  we  prove later to be sufficient to guarantee $\frac{10}7 x \in T^n$ for \emph{any} $\frac{1}{2}$-integer vector $x$ in $S^n$:  

\begin{itemize}
	\item[(A)] The support graph $G_x$ of $x$ has a Hamiltonian cycle $H$.
	\item[(B)] Vector $x$ is a convex combination of incidence vectors of $1$-trees of $K_n$ such that each $1$-tree satisfies the following condition:  it contains an even number of edges in every cut of $G_x$ consisting of four $\frac{1}{2}$-edges in $H$.
\end{itemize} 
We will use $\chi^H$ of (A) as part of the convex combination for $\frac{10}7 x$, which is globally good, since $H$ has only $n$ edges, but the $\frac{1}2$-edges of $H$ have too high a value (equal to $1$), contributing too much for the convex combination. To compensate for this, property~(B)  ensures that $x$ is not only a convex combination of $1$-trees, but these $1$-trees are even for certain edge cuts $\delta(S)$, allowing us to use a value essentially less than the  $\frac{x}{2}=\frac{1}4$ for  $\frac{1}2$ edges  in $H$ for the corresponding $T$-join.  The details of how to ensure feasibility for the $T$-join polyhedron will be given in the proof of Theorem~\ref{thm:sufficient}.

While condition (A) may look at first sight impossibly difficult to meet,  Lemma~\ref{lem:Ham} shows that one can count on the bonus of the naturally arising properties: any square point $x$ satisfies property (A), and the  additional property stated in this lemma together with the ``rainbow $1$-tree decomposition'' of Theorem \ref{thm:michel} will also imply (B) for square points. The reason we care about the somewhat technical property (B) instead of its more natural consequences is future research: in a new situation we may have to use the most general condition.   

\begin{lemma}\label{lem:B}
	Let $x$ be any square point.  Then $x$ satisfies both (A) and (B). 
\end{lemma}

\begin{proof} If we replace the $1$-paths in the support graph $G_x$ by single $1$-edges, then by Lemma~\ref{lem:Ham} there is a Hamiltonian cycle for the new graph that contains all of the $1$-edges, and thus $G_x$ also has a Hamiltonian cycle $H$ that contains all of its $1$-edges.  Thus point $x$ satisfies Property (A) by Lemma~\ref{lem:Ham}. Moreover, since $H$ contains all the $1$-edges in $G_x$, it follows that $H$ contains a perfect matching from each square of $G_x$. 
	
	Define $\Pscr$ to be the partition of the set of $\frac12$-edges of $G_x$ into pairs whose classes are the perfect matchings of squares.   Then by Theorem \ref{thm:michel}, $x$ is in the convex hull of incidence vectors of $1$-trees that contain exactly one edge from each pair $P\in\Pscr$.    Property (B) follows, since every cut $C$ that contains four $\frac12$-edges of $H$ is partitioned by two classes  $P_1,P_2\in\Pscr$. (Indeed, we saw at the end of the first paragraph of this proof  that a square met by $H$ is met in a perfect matching.) Since $P_1$ and $P_2$ are met by exactly one edge of each tree of the constructed convex combination, $C$ is met by exactly two edges of each tree.  \qed
\end{proof} 
		
		
		Next we prove that properties (A) and (B) are sufficient to guarantee that $\frac{10}7 x \in T^n$ for \emph{any} $\frac{1}{2}$-integer point of $S^n$. 
		Recall that properties (A) and (B) are more general than what we need for square points: the condition of the theorem we prove does not require that the Hamiltonian cycle of property (A) contains the $1$-edges of $G_x$, as Lemma~\ref{lem:Ham} asserts for square points. We keep this generality of (A) and (B) to remain open to eventual posterior demands of future research.   
		
		\begin{theorem}\label{thm:sufficient}
			Let $x\in S^n$ be a  $\frac{1}{2}$-integer point  satisfying properties (A) and (B).
			Then  $\frac{10}7 x \in T^n$.
		\end{theorem}
		
		\begin{proof} Let $G_x=(V_n,E_x)$ be the support graph of $x$, and let $H$ be any Hamiltonian cycle of $G_x$, which exists according to (A). Let the $1$-trees in the convex combination for property (B) be $F_i$, $i = 1, 2,...,k$, and for any $1$-tree $F$ of $G_x$ let $T_{F}$ be the set of odd degree nodes in the graph $(V_n,F)$.  Consider the vector $y\in\Rset^{E_n}$ defined as $y:=\frac{2}{3}x - \frac{1}{6}\chi^H$.
			
		\medskip
		
			\noindent
			{\bf Claim}.  For any $1$-tree $F$ of $G_x$ which satisfies the condition of property (B), vector $y$ is in the $T_{F}$-join polyhedron of $K_n$.
			\medskip
			
\noindent To prove the claim, we show that $y$ satisfies the constraints of the $T_{F}$-join polyhedron (\ref{Tjoinpoly}) of $K_n$. Clearly $y_e \ge 0$ for all $e\in E_n$.  Let $C$ be a $T_{F}$-cut in $K_n$.  We must show $y(C) \ge 1$.  Note that $|H \cap C|\ne 0$ is even, and $H\subseteq E_x$.

			\medskip\noindent
Case 1:  $|H \cap C|= 2$.   Since $x(C) \ge 2$ we have $y(C) \ge 2(\frac{2}{3}) - 2(\frac{1}{6}) = 1$, as required.  

\medskip\noindent
Case 2:   $|H \cap C| = 4$. Since $x\in S^n$, we have $x(C) \ge 2$ and since $x$ is $\frac{1}{2}$-integer the $\frac{1}{2}$-edges form edge-disjoint cycles, so  $x(C)$ is integer:  $x(C)\ge 3$, since  $x(C)=2$ would imply that $C$ consists of the four edges of $H \cap C$, and by (B) $C$ is then not a $T_{F}$-cut;  so $y(C) \ge  3(\frac{2}{3}) - 4(\frac{1}{6}) \ge 1$.

\medskip\noindent
Case 3: $|H  \cap C|\ge 6$. Then for all $e \in H  \cap C$: $y_e = \frac{2}{3} x_e - \frac{1}{6}\chi^H_e \ge\frac{2}{3}(\frac{1}{2})-\frac{1}{6}=\frac{1}{6}$, so $y(C)\ge 1$, which completes the proof of the claim.

			\medskip
		According to the claim, for all $i\in\{1,\ldots, k\}$ $y$ majorates a convex combination of $T_{F_i}$-joins, and adding any one of these to $\chi^{F_i}$, we get tours. The convex combination of these tours  is majorated by $\chi^{F_i} + y$, so  $\chi^{F_i} + y \in T^n$ for all $i=1,\ldots,k$, and therefore $x +y \in T^n$. Thus  $z:= \frac{1}{7}\chi^H + \frac{6}{7}(x + y)= \frac{10}{7} x$ is also in $T^n$, which completes the proof.  \qed
		\end{proof}

		\medskip
		Our main result is an immediate corollary of this theorem: 
		
		
		
		\begin{theorem}\label{thm:main}
			Let $x$ be a square point. Then $\frac{10}{7}x \in T^n$. 
		\end{theorem}
		
		\begin{proof} By Theorem~\ref{thm:sufficient} it is enough to make sure that $x$ satisfies properties (A) and (B),  
			which is exactly the assertion of Lemma~\ref{lem:B}. 
			\qed
		\end{proof}
		
		
		Since $\frac{1}2$-integer Boyd-Carr points are square points we have: 
		
			\begin{corollary}\label{cor:BoydCarr}
		If $x$  is a $\frac{1}2$-integer Boyd-Carr point, then   $\frac{10}{7}x \in T^n$.
		\end{corollary}

		Theorem~\ref{thm:main} immediately implies the following optimization form of the above corollary: 
		
		\begin{corollary}\label{cor:approxalg}
			Let  $c\in \Rset^{E_n}$ be a metric cost function optimized by a square point $x$, i.e. $cx = LP(c)$.  Then there exists a Hamiltonian cycle of cost at most $\frac{10}{7} LP(c)$  in $K_n$. 
		\end{corollary}
	
	\begin{proof} We have $\frac{10}{7} LP(c)=\frac{10}{7} cx= c(\frac{10}{7} x)\ge OPT_{T^n}(c)$, where the last inequality follows from $\frac{10}{7} x$ being a convex combination of tours by Theorem~\ref{thm:main}.  As  $OPT_{T^n}(c) \ge OPT(c)$ for metric costs, the result follows. \qed	
	\end{proof}

In the following section we show that the above existence theorems and their corollaries can also be accompanied   with polynomial algorithms. We finish this section by showing that the integrality gap of square points is at least $\frac{4}3$, providing new  examples showing that  Conjecture \ref{conj:4over3} cannot be improved.  

Define a subclass of square points we call \emph{$k$-donuts}, $k\in \mathbb{Z}$, $k\ge 2$: the support graph $G_x=(V_n,E_x)$ consists of $k$ $\frac{1}2$-squares arranged in a cyclic order, where the squares are joined by $1$-paths, each of length $k$.  
		In Figure \ref{fig:donut} the support graph of a $4$-donut is shown.  In the figure, dashed edges represent $\frac{1}2$-edges and solid edges represent $1$-edges.
		
		We define the cost of each edge in $E_x$ to be $1$, except for the pair of $\frac12$-edges in each square that are opposite to one another, with one edge on the inside of the donut, and one on the outside, which are defined to have cost $k$ (see the figure, where only edges of cost $k$ are labelled, and all other edges in $E_x$ have cost $1$). The costs of other edges of $K_n$ not in $E_x$ are defined by the metric closure (cost of shortest paths in $G_x$). 
		For these defined costs $c^{(k)}$, we have $OPT(c^{(k)})= 4k^2 - 2k + 2$ and $LP(c^{(k)}) \le c^{(k)}x = 3k^2 + k$, thus  $\lim_{k\rightarrow\infty}\frac{OPT(c^{(k)})}{LP(c^{(k)})}\ge \frac{4}3$.  Along with Theorem \ref{thm:main}, this gives the following:
		
		\begin{corollary}\label{cor:intgapbds}
			The integrality gap for square points lies between $\frac{4}3$ and $\frac{10}7$. 
		\end{corollary}

		\begin{figure}
			\centering
			\includegraphics[width=.4\textwidth]{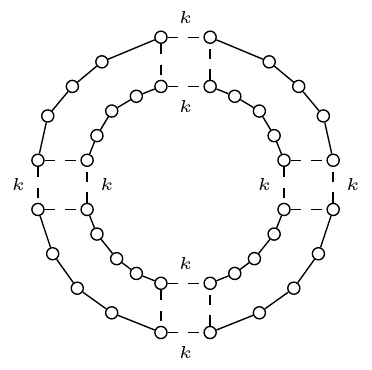}
			\caption{Graph $G_x$ for a $k$-donut $x$, $k$=4.}
			\label{fig:donut}
		\end{figure}

		To conclude this section we briefly discuss the structure of  Carr-Vempala points.   
		Note that for the Boyd-Carr points that have been our focus, the transformation used from general vertices $x\in S^n$ to these Boyd-Carr points does not completely preserve the denominators.  In particular, $\frac{1}2$-integer vertices of $S^n$ get transformed into Boyd-Carr points $x^*$ with $x^*_e$ values in $\{1, \frac{1}2, \frac{3}4, \frac{1}4, 0\}$.
		However, for the Carr-Vempala points, it is clear from their construction in \cite{CV} that general $\frac{1}2$-integer vertices of $S^n$ lead to $\frac{1}2$-integer Carr-Vempala vertices. In fact we have the following theorem which, if Conjecture \ref{conj:half} is true, would provide a nice approach for proving Conjecture \ref{conj:fourthird}, since it is given for free that Carr-Vempala vertices satisfy property (A):
		\begin{theorem}\label{CarrVempalahalf}
			If $\rho x \in T^n$ for each $\frac{1}2$-integer Carr-Vempala point $x\in S^n$, then   $\rho x \in T^n$ for every $\frac{1}2$-integer point $x \in S^n$.
		\end{theorem}
		In light of these results and conjectures it seems worthwhile to study fundamental classes further and the role of $\frac12$-integer points in them.
		
		\section{Related Algorithms}\label{sec:delta}
		
		In this section we show that some of the existence theorems stated in the previous sections can be accompanied by simple combinatorial algorithms that can be executed in polynomial time.  The main result of this section is a  polynomial-time algorithm for finding the Hamiltonian cycle of Corollary~\ref{cor:approxalg} in a completely elementary way.
		We point out that the bridge taking us to this result also leads to a polynomial-time algorithm for finding a minimum cost Hamiltonian cycle containing the $M$ edges in a square graph, generalizing Lemma~\ref{lem:Ham}. 
		
		It turns out that the greedy algorithm is the main ingredient of our algorithms, and delta-matroids are the structure behind this phenomenon. We give a short introduction to delta-matroids, their greedy algorithm and their connection to our results.  We first introduce the algorithm directly on our combinatorial objects below. 
		
	\medskip
		\noindent 
		{\bf Greedy Algorithm (HAM) for Hamiltonian cycles in square graphs }
		
		\smallskip\noindent
		{\bf Input}: A square graph $(G,M)$ and cost vector $c\in\Rset^{E(G)}$.
		
		\smallskip\noindent
		{\bf Output}: A Hamiltonian cycle containing $M$.

		\smallskip\noindent
		1. For each square $C$ of $G$ let $w_C$ be the absolute value of the difference of the cost of the two perfect matchings of $C$. Order the squares in non-increasing order of $w_C$, from $C_1$ to $C_t$; $i:=0$.  
		
		\smallskip\noindent
		2. i:=i+1 ; while $i\le t$ do: 
		
\setlength{\leftskip}{.7cm}	
 \noindent We keep exactly one of the two perfect matchings of $C_i$ and delete both edges of the other perfect matching according to the following rule:
 			
\setlength{\leftskip}{1.2cm}
\noindent If the graph remains connected after both of these choices, keep the perfect matching of $C_i$ which has smaller cost (if the costs are equal, break ties arbitrarily). 
			
\noindent If the graph remains connected after exactly one of these two possible choices, then make this choice.

\setlength{\leftskip}{.7cm}
\noindent There is no other case according to the following claim:
			
\medskip			
\noindent{\bf Claim:}  The graph remains connected with at least one of the two choices.

\smallskip\noindent{\bf Proof}: Suppose  that at iteration $i$, $G-C_i$ (edge-deletion) is not connected. Then it  has at least two of the four arising nodes of degree $1$ (the only nodes of odd degree after contracting the remaining squares) in each component. It  follows that it has two components. If adding one of the two perfect matchings it is still not connected, then all edges of $C_i$ are induced by one of the two components, so  $G$ is also not connected, a contradiction. \qed
\setlength{\leftskip}{0pt}			
\noindent 3. Output the constructed graph which is a Hamiltonian cycle containing $M$ (all degrees are two, and it is connected because of the claim).
		
		\noindent{\bf end}

		\medskip
		     We will see that this algorithm determines  the Hamiltonian cycle of minimum cost containing $M$ (Theorem~\ref{thm:optham}), and this is also not difficult to prove directly along the same lines as the optimality of the greedy algorithm for optimal spanning trees, as follows:   
		Suppose for a contradiction that the algorithm finds $H$, while $K$ is a Hamiltonian cycle  of smaller cost containing $M$.  
	Then there exists a square $C_i$ for which $K$ uses a different perfect matching than $H$ and the cost of this perfect matching is strictly less than the one used by $H$. 
		Let $i$  be the smallest index for which this is true, and assume that over all minimum cost Hamiltonian cycles containing $M$, we chose $K$ to be the one for which this $i$ is as large as possible. By the algorithm, when we considered square $C_i$, we know that removing the smaller cost perfect matching disconnected the graph (or it would have been chosen). Thus there must exist another square $C_j$ crossing the cut formed by this disconnection for which $K$ chooses a different perfect matching than $H$, and $j < i$.  By choice of index $i$, we know that the perfect matching used in $C_j$ by $K$ has cost greater than or equal to the one used by $H$.  Now consider the new Hamiltonian cycle $K^\prime$ obtained by taking $K$, and swapping the perfect matchings used by $K$ in squares $C_i$ and $C_j$.  We have $c(K^\prime) = c(K) + w_{C_i} - w_{C_j}$ and since $w_{C_j} \ge w_{C_i}$ by the index ordering, we must have $K^\prime$ is another minimum cost Hamiltonian cycle.  Again there must exist a square $C_r$ for which $K^\prime$ uses a different perfect matching than $H$ and the cost of this perfect matching is strictly less than the one used by $H$, but by construction of $K^\prime$ we have $r>i$.  But this contradicts our choice of $K$.
		
		
		
		We now propose the following algorithm for finding a tour of relatively small cost for cost functions optimized at square points. 
		
			\medskip
		\noindent 
		{\bf Algorithm (TOUR) for  tours in the case of a square optimum for the subtour LP }
		
			\smallskip\noindent
		{\bf Input}: Costs $c\in\Rset^{E_n}$ and a square point $x$ optimizing $c$ on the subtour LP, i.e. $c(x)= LP(c)$.
		
		\smallskip\noindent
		{\bf Output}: A tour in $G_x$.

		\smallskip\noindent
		1. Determine the support graph $G_x$, and call (HAM) for the square graph $(G,M)$  that results from replacing  each $1$-path of $G_x$ by one single edge whose cost is the sum of the costs of the replaced edges, and defining $M$ to be the set of these single edges.   Let $H$ be the Hamiltonian cycle of $G_x$ obtained by taking the output of (HAM) and replacing all edges of $M$ by their respective $1$-paths. 
		
			\smallskip\noindent
		2. 	As in Theorem~\ref{thm:michel} and Lemma \ref{lem:B}, determine the partition $\Pscr$ of the $\frac{1}{2}$-edges of $G_x$ into pairs whose classes are the perfect matchings of the squares of $G_x$, and find the spanning tree
		$F^*$ of $G_x$ of minimum cost having exactly one edge in each $P\in\Pscr$ in polynomial time with Edmonds matroid intersection algorithm \cite{EdmondsMIpolytope}. 
		
				\smallskip\noindent
		3.  Let $T_{F^*}$ be the set of odd degree nodes in the graph $(V_n,F^*)$.  Find a minimum cost $T_{F^*}$-join in $G_x$. Note this can be done in polynomial time (cf. \cite{Schrijver}).		
		
			\smallskip\noindent
		4. Let $J^*$ be the union of $F^*$ and the $T_{F^*}$-join from Step 3, and output the one of $H$ or $J^*$ having smaller cost.   
		

		\smallskip
		We can now complete Corollary~\ref{cor:approxalg} with an algorithmic postulate.

		\begin{theorem}\label{thm:approxalg}
			Let  $c\in \Rset^{E_n}$ be a metric cost function optimized by a square point $x$, $cx = LP(c)$.  Then there exists a Hamiltonian cycle of cost at most $\frac{10}{7} LP(c)$  in $K_n$  and (TOUR) can be used to find such a cycle in polynomial time. 
		\end{theorem}
		
		\begin{proof} We have to prove only the latter part of the last sentence concerning the algorithm, as the rest has already been included in Corollary~\ref{cor:approxalg}.

			Let $H$ be any Hamiltonian cycle of the support graph $G_x$ that contains all of the $1$-edges, provided algorithmically in polynomial time by (HAM), and let $J^*$ be the tour of $G_x$ determined in polynomial time by (TOUR).   
			 Note that the $1$-tree $F^*$ from Step 2 of (TOUR) satisfies the condition of property (B), and thus by the Claim in the proof of Theorem \ref{thm:sufficient} we have
			 $y:=\frac{2}{3}x - \frac16\chi^H$ is in the $T_{F^*}$-join polyhedron (\ref{Tjoinpoly}) of $K_n$ . Thus the cost of the minimum cost $T_{F^*}$-join found in Step 3 is at most $c(y)$.   Similarly, by Theorem \ref{thm:michel}, $c(F^*)\le c(x)$, which gives
	$c(J^*)\le c(x) + c(y)$. Thus 
	\[\min\{c(H), c(J^*)\}\le  \frac{1}{7}c(H) + \frac{6}{7}c(J^*)\le  \frac{1}{7}c(H) + \frac67 c(x) + \frac67 (\frac23 c(x) -  \frac16 c(H))=
 \frac{10}{7}LP(c)\]  is the cost of a tour, and shortcutting this tour we get a Hamtiltonian cycle of $K_n$ not larger in cost.  
	\qed
		\end{proof}

		

		
		
		Note that this proof actually used less than what (HAM) produces: for TOUR it is sufficient to find {\em any} Hamiltonian cycle in $G_x$ containing its  $1$-edges, not necessarily the optimal one! However, sharper or more general results may need the exact optimum here. This motivates us to sketch some details about delta-matroids that are in the background.  
		
		Delta-matroids were introduced by Bouchet \cite{B93}. For the introduction and the basics about them we follow \cite{BC95}: the family $\Dscr\ne\emptyset$ of subsets of a finite set $S$, or the pair $(S,\Dscr)$ is called a {\em delta-matroid} if the following symmetric exchange axiom (also called the $2$-step axiom) is satisfied: For $D_1, D_2\in\Dscr$ and $j\in D_1\Delta D_2$ there exists $k\in D_1\Delta D_2$ such that $D_1\Delta\{j,k\}\in\Dscr$. Note that $k=j$ is possible, and we then naturally define $\{j,j\}:=\{j\}$.

		A delta-matroid $(S,\Dscr)$ may have an exponential number of elements, too many to be given explicitly. Fortunately, in order to work with delta-matroids we need less than giving all of its elements as input.  
		The basic and simple greedy algorithm already necessitates a solution of the following problem  \cite[(2.1)]{BC95}: let  $(S,\Dscr)$ be a delta-matroid, then for given $A, B\subseteq S$ decide  whether there exists $D\in\Dscr$ such that $D\supseteq A$, $D\cap B=\emptyset$. Let us call an oracle that solves this problem the {\em extendability oracle} for the given delta-matroid. This oracle can be executed in polynomial time for all the relevant applications, and we will have to check that this also holds for the delta matroid for which we need it. 
		
		Think about the set $A$ as a set of elements chosen to be in the solution, and $B$ the set of elements chosen not to be in the solution. Roughly, for an objective function $c\in\Rset^S$, the greedy algorithm considers the elements  of $S$ in decreasing order of the absolute values and attempts to put a considered $s\in S$  into $A$ if $c(s)\ge 0$  and ``it is possible'' to be put in $A$,  and to $B$ if  $c(s)\le 0$ and ``it is possible'' to put it in $B$, where ``it is possible'' means a YES answer of the extendibility oracle with the attempted update of $A$ and $B$ (see precisely below). 
		
		The following theorem is a generalization of  Lemma~\ref{lem:Ham} and therefore a direct proof of it.  Given a square graph $(G,M)$, let  $R=R(G)\subseteq E(G)\setminus M$ be a {\em reference set} containing  exactly one edge from each square of  $E(G)\setminus M$; $|R|=|V(G)|/4$. Let $\Hscr=\Hscr(G)$ be the set of Hamiltonian cycles of $G$ containing $M$.

		\begin{theorem}\label{thm:ham} Let $(G,M)$ be a square graph. Then $\Hscr\ne\emptyset$, and \[\Dscr:=\{H\cap R:  H \hbox{ is a Hamiltonian cycle of $G$ containing $M$}\}\] is a delta-matroid. 
		\end{theorem}
		
		\begin{proof} To prove the first part, replace one by one, one after the other in arbitrary order, the squares by one of their two perfect matchings. 
			
			\smallskip\noindent
			{\bf Claim}: With at least one of the two choices  the graph remains connected. 
			
			\smallskip
			It follows that at the end of the procedure we get a connected graph with all degrees equal to $2$ containing $M$, that is, a  Hamiltonian cycle containing $M$.  
			
			\medskip 
			To prove the claim note that the graph $G_C$ we get after deleting the square $C$  has at most $2$ components. So one of the two perfect matchings of $C$ must join the two components, since otherwise the graph obtained by adding back $C$ is also not connected. The claim is proved.

			\medskip So we proved that $\Dscr\ne\emptyset$. In order to prove that it is a delta-matroid, we have to check that {\em for any two Hamiltonian cycles $H_1\ne H_2$, and square $C$ of $G-M$ where $H_1$ and $H_2$ do not use the same perfect matching of $C$, either $H_1\triangle C$ is also a Hamiltonian cycle or there exists a square $D$ of $G-M$ 
				so that $H_1\triangle C\triangle D$ is a Hamiltonian cycle.}

			To prove this,  suppose $H_1\triangle C$ is not a Hamiltonian cycle. It is still a {\em $2$-factor} -- subset of edges with all degrees equal to $2$ --  with two components, with the cut $Q\subseteq E(G)$ between the two components. Since $H_2$ is connected, it contains a square $D$ so that for one of the two  perfect matchings of  $D$:  $H_2\cap D \cap Q\ne\emptyset$. But then clearly, $H_1$ and $H_2$ do not use the same perfect matching of $D$, and $H_1\triangle C \triangle D$ is again connected, and thus a Hamiltonian cycle containing $M$.\qed
		\end{proof}
		
		Let us call the delta-matroid $\Dscr$ of the theorem {\em square}. It is the same delta-matroid as Bouchet's delta-matroid of transitions in Eulerian trails \cite{B93}.
		
		\begin{lemma}
			For square delta-matroids the extendability oracle can be computed in polynomial time. 
		\end{lemma}
		
		\begin{proof} Assume $A, B\subseteq R$. If $A\cap B\ne\emptyset$ the answer of the extendability oracle is NO. So let $A\cap B=\emptyset$.
			
			For edges  $a\in A$, choose in the square of $a$ the perfect matching containing $a$, and for $b\in B$ in the square of $b$ the one that does not contain $b$. 
			
			If the obtained graph is not connected, then clearly,   the extendability oracle gives a NO answer. In case it is  connected, replace each of the remaining squares (those disjoint from $A\cup B$) one after the other by  one of its two perfect matchings, so that it remains connected.  One of the two choices does indeed keep the graph connected, since if not, adding both perfect matchings, (like in the proof of Theorem~\ref{thm:ham}) it would also not be connected.\qed  
		\end{proof}

\noindent 
		{\bf General Greedy Algorithm (GREEDY) for Delta-Matroids} \cite{BC95}
		
		\smallskip\noindent
		{\bf Input}: Delta-matroid $(S,\Dscr)$ given with the extendability oracle, and cost vector $c\in\Rset^S$.
		
		\smallskip\noindent
		{\bf Task}: Determine $D\in\Dscr$, equivalently the vector $\chi_D$, of minimum cost.

		\smallskip\noindent
		1. Order the $S$ decreasingly in the absolute values  of $c$, that is, $|c_1|\ge |c_2|\ge \ldots |c_{|S|}|,$ where we can suppose $S=\{1,\ldots,{|S|}\}$. Define $i:=0$, $A_0:=B_0:=\emptyset$. 

\noindent
		2. i:=i+1 ; while $i\le n$ do: 
		\begin{itemize}
			\item[-] If $c_i \le 0$: 
			
			In case the extendability oracle gives a YES answer with input $A_i:=A_{i-1}\cup\{i\}$ and $B_i:=B_{i-1}$, then keep this definition of $A_i$, $B_i$.
			
			In case the extendability oracle gives a NO answer, $A_i:=A_{i-1}$, $B_i:=B_{i-1}\cup\{i\}$. (Since $(A_i,B_i)$ is extendable, the answer for the latter choice is  YES in this case.) 
			
			\item[-] If  $c_i >0$:
			
			In case the extendability oracle gives a YES answer with input $A_i:=A_{i-1}$ and $B_i:=B_{i-1}\cup\{i\}$, then keep this definition of $A_i$, $B_i$. 
			
			In case the extendability oracle gives a NO answer, $A_i:=A_{i-1}\cup\{i\}$, $B_i:=B_{i-1}$.
			\end{itemize}  		
		\noindent
			3. Output $D:=A_n$.
			
			\noindent{\bf end}

			\medskip 
			It can be readily checked that $\{A_n, B_n\}$ {\em is a partition of $S$, $A_n\in\Dscr$, and it is also not difficult to check that   $c(A_n)=\min \{c(D): D\in \Dscr\}$.}  Moreover, (HAM) is a special case,  and we have:
			
			\begin{theorem}\label{thm:optham}
			The output of (HAM) for a square graph $(G,M)$ is a Hamiltonian cycle of $G$ containing $M$ of minimum cost.
			\end{theorem}


		
		
		
		\section*{Acknowledgements}
We are indebted to Michel Goemans for an email from his sailboat with a pointer to Theorem \ref{thm:michel}; to Alantha Newman, Frans Schalekamp, Kenjiro Takazawa and Anke van Zuylen for helpful discussions. 


\begin{thebibliography}{99}
%
%
\bibitem{BB} G. Benoit, S. Boyd, Finding the exact integrality gap for small traveling salesman problems, Mathematics of Operations Research 33(4), 921--931 (2008)

\bibitem{B93} A. Bouchet, Compatible Euler Tours and Supplementary Eulerian Vectors, Europ. J. Combinatorics 14, 513--520 (1993)  

\bibitem{BC95} A.~Bouchet, W.~Cunningham, Delta-matroids, Jump Systems, and Bisubmodular Polyhedra, SIAM Journal on Discrete Mathematics 8(1), 17--32 (1995)


\bibitem{BC} S. Boyd, R. Carr, Finding low cost TSP and 2-matching solutions using certain half-integer subtour vertices,  Discrete Optimization 8,  525--539 (2011)

\bibitem{BIT} S.~Boyd, S.~Iwata, K.~Takazawa, Finding $2$-factors closer to TSP tours in cubic graphs, SIAM Journal on Discrete Mathematics,  27(2), 918--939 (2013)

\bibitem{BL} H. Broersma, X. Li, Spanning trees with many or few colors in edge-colored graphs, Discussiones Mathematicae Graph Theory  17,  259--269 (1997)

\bibitem{CR} R. Carr, R. Ravi, A new bound for the $2$-edge connected subgraph problem, Proceedings of Integer Programming and Combinatorial Optimizaiton (IPCO), Lecture Notes in Computer Science, Springer,  112--125 (1998)

\bibitem{CV} R. Carr, S. Vempala, On the Held-Karp relaxation for the asymmetric and symmetric travelling salesman problem, Mathematical Programming A 100, 569--587 (2004) 

\bibitem{chr} N. Christofides, Worst case analysis of a new heuristic for the traveling salesman problem, Report 388, Graduate School of Industrial Administration, Carnegie-Mellon University, Pittsburgh, PA, 1976.

\bibitem{cunningham} W. H. Cunningham, On bounds for the metric TSP, manuscript, School of Mathematics and Statistics, Carleton University, Ottawa, Canada, 1986.

\bibitem{EdmondsMIpolytope}
J.~Edmonds,
Submodular functions, matroids and certain polyhedra.
In: Combinatorial Structures and Their Applications;
Proceedings of the Calgary International Conference on Combinatorial Structures and Their Applications 1969
(R.~Guy, H.~Hanani, N.~Sauer, J.~Sch\"onheim, eds.),
Gordon and Breach, New York, 1970.

\bibitem{EdmondsJ73} J. Edmonds, E.~L. Johnson,
Matching, euler tours and the chinese postman,
Mathematical Programming 5(1),  88--124 (1973)


\bibitem{G}
S.O.~Gharan,  A.~Saberi, M.~Singh, 
A randomized rounding approach to the traveling salesman problem,
Proceedings of the 52nd Annual IEEE Symposium on Foundations of Computer Science, 550--559 (2011)

\bibitem{goemans}
M.X. Goemans, Worst-case comparison of valid inequalities for the TSP, Mathematical Programming 69, 335--349 (1995)

\bibitem{JC}
C.~Gottschalk, J.~Vygen, Better {\it s}-{\it t} -tours
by Gao trees, Mathematical Programming 172(1-2), 191--207 (2018)

\bibitem{theTSPbook} M.~Gr\"otschel, M.~W.~Padberg, ``Polyhedral Theory'', in E.L. Lawler, J.K. Lenstra, A.H.G. Rinnooy Kan, D.B. Shmoys, eds., The Traveling Salesman Problem -- A Guided Tour of Combinatorial Optimization, Wiley, Chichester, 1985.

\bibitem{HNR} A. Haddadan, A Newman, R. Ravi, Shorter tours and longer detours: Uniform covers and a bit beyond, arXiv 1707.05387v3 [cs.DS], 2017.

\bibitem{KS}  T.~Kaiser, R.\v{S}krekovski, Cycles intersecting edge-cuts of prescribed sizes, SIAM Journal on Discrete Mathmatics 22(3), 861--874 (2008) 

\bibitem{Kotzig} A. Kotzig, Moves without forbidden transitions in a graph, Mat. Casopis Sloven, Akad. Vied 18,  76--80 (1968)


\bibitem{SWZ} F. Schalekamp, D. Williamson, A. van Zuylen, 2-Matchings, the Traveling Salesman Problem, and the Subtour LP: A Proof of the Boyd-Carr Conjecture, Math. Oper. Res. 39(2), 403--417 (2014) 

\bibitem{Schrijver} A.~Schrijver, Combinatorial Optimization, Springer-Verlag Berlin Heidelberg, 2003.

\bibitem{Sebo14}
A. Seb\H{o}, Y. Benchetrit, M. Stehlik,  Problems about uniform covers, with tours and detours, Matematisches Forschungsinstitut Oberwolfach Report No.~51/2014, DOI: 10.4171/OWR/2014/51,  2912--2915 (2015) 

\bibitem{AndrasJens}  A. Seb\H{o}, J. Vygen, Shorter tours by nice ears, Combinatorica 34,  597-629 (2014)

\bibitem{AA} A.~Seb\H{o}, A.~van Zuylen, The Salesman's Improved Paths: A 3/2+1/34 Approximation, Foundations of Computer Science (FOCS), 118-127 (2016) 


\bibitem{SW} D. Shmoys, D. Williamson, Analysis of the Held-Karp TSP bound:  A monotoncity property with application, Information Processiong Letters 35,  281--285 (1990)

\bibitem{Wolsey80} L. Wolsey, Heuristic analysis, linear programming and branch and bound, Mathematical Programming Study 13,  121--134 (1980)

\end{thebibliography}


\end{document}